\documentclass{acm_proc_article-sp}

\usepackage{amsmath}
\usepackage{amssymb}
\usepackage{algorithm}
\usepackage{algorithmic}
\usepackage{array}
\usepackage{complexity}
\usepackage{url}

\providecommand*{\Zset}{\mathbb{Z}} 
\providecommand*{\Qset}{\mathbb{Q}} 
\providecommand*{\Rset}{\mathbb{R}} 

\newcommand*{\bigland}{\mathop{\bigwedge}}

\newtheorem{theorem}{Theorem}[section]
\newtheorem{proposition}[theorem]{Proposition}
\newtheorem{lemma}[theorem]{Lemma}
\newtheorem{corollary}[theorem]{Corollary}
\newtheorem{definition}[theorem]{Definition}
\newtheorem{example}[theorem]{Example}
\newtheorem{remark}[theorem]{Remark}

\newcommand{\st}{\mathrel{.}}
\newcommand{\itc}{\mathrel{:}}

\newcommand{\dec}{\mathrm{DEC}}
\newcommand{\pos}{\mathrm{POS}}
\newcommand{\inc}{\mathrm{INC}}

\begin{document}

\title{Eventual Linear Ranking Functions}

\numberofauthors{2}
\author{
\alignauthor
Roberto Bagnara \\
    \affaddr{BUGSENG (\url{http://bugseng.com})} \\
    \affaddr{Dipartimento di Matematica e Informatica} \\
    \affaddr{Universit\`a  di Parma, Italy} \\
    \email{bagnara@cs.unipr.it}
\alignauthor
Fred Mesnard \\
    \affaddr{LIM/ERIMIA}\\
    \affaddr{Universit\'e de la R\'eunion, France} \\
    \email{frederic.mesnard@univ-reunion.fr}
}

\maketitle

\begin{abstract}
Program termination is a hot research topic in program analysis.
The last few years have witnessed the development of termination
analyzers for programming languages such as C and Java with remarkable
precision and performance.
These systems are largely based on techniques and tools coming
from the field of declarative constraint programming.
In this paper,%
\footnote{A preliminary version of this work, in French, has been
    presented to the \emph{Journ\'ees Francophones de Programmation par
    Contraintes}.
}
we first recall an algorithm based on Farkas' Lemma for
discovering linear ranking functions proving termination of a certain
class of loops.  Then we propose an extension of this method for
showing the existence of \emph{eventual linear ranking functions},
i.e., linear functions that become ranking functions after a
finite unrolling of the loop.
We show correctness and completeness of this algorithm.
\end{abstract}


\section{Introduction}
\label{sec:introduction}

Program termination is a hot research topic in program analysis.
The last few years have witnessed the development of termination
analyzers for mainstream programming languages such as C~\cite{CookPR06}
and Java~\cite{AlbertAGGPRRZ09,OttoBEG10,SpotoMP10} with remarkable
precision and performance.
These systems are largely based on techniques and tools coming
from the field of declarative constraint programming.

Beyond the specificities of the targeted programming languages and after
several abstractions (see, e.g., \cite{SpotoMP10}), termination analysis
of entire programs boils down to termination analysis of individual loops.
Various categories of loops have been identified: for the purposes of
this paper we focus on \emph{single-path linear constraint} (SLC)
loops~\cite{Ben-AmramG13TR}.
An SLC loop over $n$ variables $x_1$, \dots,~$ x_n$ has the form
\[
  \textbf{while} \; (B\ \mathbf{x} \le \mathbf{b})
  \quad \textbf{do} \;
    A \begin{pmatrix}
        \mathbf{x} \cr
        \mathbf{x}'
      \end{pmatrix} \le \mathbf{c}
\]
where $\mathbf{x} = (x_1, \ldots, x_n)^T$
and $\mathbf{x}'=(x'_1, \ldots, x'_n)^T$
are column vectors of variables,
$B \in \Zset^{p \times n}$ is an integer matrix, $\mathbf{b} \in \Zset^p$,
$A \in \Zset^{q \times 2n}$ and $\mathbf{c} \in \Zset^q$.
Such a loop can be conveniently written as a constraint logic
programming rule:
\[
   p(\mathbf{x})
     \leftarrow
       B \mathbf{x} \le \mathbf{b}, \;
       A \begin{pmatrix}
           \mathbf{x} \cr
           \mathbf{x}'
         \end{pmatrix} \le \mathbf{c}, \; p(\mathbf{x}').
\]
When variables take their values in $\Zset$ (resp., $\Qset$),
we call such loops \emph{integer} (resp., \emph{rational}) loops.
They model a computation that starts from a point $\mathbf{x}$;
if $B\ \mathbf{x} \le \mathbf{b}$ is false, the loop terminates;
otherwise, a new point $\mathbf{x}'$ is chosen that satisfies
\(
   A (\begin{smallmatrix}
        \mathbf{x} \cr
        \mathbf{x}'
      \end{smallmatrix}) \le \mathbf{c}
\)
and iteration continues replacing the values of $\mathbf{x}$
by those of $\mathbf{x}'$.

Loop termination can always be ensured by a \emph{ranking function}
$\rho$, a function from $\Zset^n$ or $\Qset^n$ to a well-founded set.
As the domain of $\rho$ is well-founded, the computation terminates.
To the best of our knowledge, decidability of universal termination
of SLC loops (i.e., from any starting point and for any choice of
the next point at each iteration) is an open question.
Some sub-classes have been shown to be
decidable~\cite{BozgaIK12,Braverman06,Tiwari04}.
For instance, Braverman proves that termination of
loops where the body is a \emph{deterministic} assignment
$\mathbf{x}' \leftarrow A \mathbf{x}$ is decidable when the variables range over $\Qset$.
The problem is open for the non-deterministic case, as stated in his paper.
On the other hand, various generalizations have been shown
to be undecidable~\cite{Ben-AmramGM12}.

A way to investigate loop termination is to restrict the class of
considered ranking functions.
In the following section, we recall a well-known technique
for computing linear ranking functions for rational SLC loops.

In Section~\ref{sec:eventual-linear-ranking-functions}
we present the main contribution of the paper,
namely the definition of \emph{eventual linear ranking functions}:
these are linear functions that become
ranking functions after a finite unrolling of the loop.
We shall see that the number of unrolling is not pre-defined, but
depends on the data processed by the loop.
Section~\ref{sec:eventual-linear-ranking-functions} presents
complete decision procedures for the existence of
eventual linear ranking functions of SLC loops.  The presentation is
gradual and illustrates the algorithms by means
of constraint logic programming (CLP) technology and dialogs with real
CLP tools.
Section~\ref{sec:related-work-and-experiments} discusses related
work and a preliminary experimentation conducted on the benchmarks
proposed in two very recent papers.
Section~\ref{sec:conclusion-and-future-work} concludes the paper.

\section{Linear Ranking Functions}
\label{sec:linear-ranking-functions}

We first define the notion of linear (resp., affine) ranking function
for an SLC loop.

\begin{definition}
\label{def-fn-rng-lin}
Let $C$ be the SLC loop
$p(\mathbf{x}) \leftarrow c(\mathbf{x}, \mathbf{x}'),$ $p(\mathbf{x}')$
where $p$ is an n-ary relation symbol.
A \emph{linear} (resp., \emph{affine}) \emph{ranking function} $\rho$ for $C$
is a linear (resp., affine) map from $\Qset^n$ to $\Qset$ such that
\[
   \forall \mathbf{x}, \mathbf{x}'
     \itc c(\mathbf{x}, \mathbf{x}')
            \implies
              \rho(\mathbf{x}) \ge 1 + \rho(\mathbf{x}')
              \,\land\,
              \rho(\mathbf{x}) \ge 0.
\]
\end{definition}

In words, continuation of the iteration, i.e., $c(\mathbf{x}, \mathbf{x}')$,
entails that $\rho$ stays positive and strictly decreases by at least $1$
for each iteration.
We point out that if $c(\mathbf{x}, \mathbf{x}')$ is not satisfiable,
the loop ends immediately and any linear function is a ranking
function.
In the paper, we assume that $c(\mathbf{x}, \mathbf{x}')$ is
satisfiable.

\begin{remark}
\textup{Definition~\ref{def-fn-rng-lin}} might seem too restrictive
when working with rational numbers as one might prefer to replace the
decrease by $1$ by a decrease by $\varepsilon$, a fixed positive
quantity.  Actually, by multiplying such an $\varepsilon$-decrease
ranking function by $1/\varepsilon$, we see that the two definitions
are equivalent with respect to the existence of a ranking function.
\end{remark}

\begin{remark}
Although the class of affine ranking functions subsumes the class of linear
ranking functions, any decision procedure for the existence of
linear ranking functions can be extended to a decision procedure
for the existence of affine ranking functions.
To see this, note that an affine ranking function for
\begin{align*}
  p(\mathbf{x})
    &\leftarrow
      c(\mathbf{x}, \mathbf{x}'), \; p(\mathbf{x}')
\intertext{%
is a linear ranking function for
}
  p(\mathbf{x}, y)
    &\leftarrow
      c(\mathbf{x}, \mathbf{x}'), y = 1, y' = 1, \; p(\mathbf{x}', y'),
\end{align*}
where $y$ is distinct from the variables in $\mathbf{x}$.
\end{remark}

In this section, we focus on linear ranking functions for SLC loops.
After the presentation of a formulation of Farkas' Lemma
we consider the problem of verifying linear ranking
functions, and then the detection of such ranking functions.

\subsection{Farkas' Lemma}

A linear inequation $I$ over rational numbers is a logical consequence
of a finite satisfiable conjunction $S$ of linear inequations when
$I$ is a linear positive combination of the inequations of $S$.
More formally, let $S$ be
\[
  \left\{
    \begin{array}{c@{\hspace{2ex}}c@{\;}c@{\hspace{2ex}}c@{\;}c@{\hspace{2ex}}c@{\;}c@{\hspace{2ex}}c@{\;}c@{\hspace{2ex}}}
      a_{1,1}x_1 & + & \cdots & + & a_{1,n}x_n & + & b_1    & \ge & 0 \\
      \cdots    & + & \cdots & + & \cdots    & + & \cdots & \ge & 0 \\
      a_{m,1}x_1 & + & \cdots & + & a_{m,n}x_n & + & b_m    & \ge & 0.
    \end{array}
  \right.
\]
and suppose that $S$ has at least one solution.
Farkas' Lemma states the equivalence of
\[
  \forall x_1, \dots, x_n
    \itc S \implies (c_1 x_1 + \cdots + c_n x_n + d \ge 0 )
\]
and
\begin{multline*}
  \exists \lambda_1 \ge 0, \ldots, \lambda_m \ge 0
    \st \\
      \Bigl( d \ge \sum\nolimits_{i=1}^m \lambda_i b_i \Bigr)
      \land
      \bigland_{j=1}^n
        \Bigl( c_j = \sum\nolimits_{i=1}^m \lambda_i a_{i,j} \Bigr).
\end{multline*}

\subsection{Verification}
\label{sec:lrf-verification}

Given an SLC loop $C$ and a linear function $\rho$,
we can easily check whether $\rho$ is a ranking function for $C$
by testing the unsatisfiability of
$c(\mathbf{x}, \mathbf{x}'), \rho(\mathbf{x}) <  1 +  \rho(\mathbf{x}')$ and
$c(\mathbf{x}, \mathbf{x}'), \rho(\mathbf{x}) < 0$.
This test has polynomial complexity and can be done with a complete
rational solver such as , e.g., CLP($\Qset$) \cite{Holzbaur95}.

\begin{example}
\label{ex:linear-ranking-function}
For the SLC loop $C$:
\[
   p(x, y) \leftarrow x \ge 0, y' \le y - 1, x' \le x + y, y \le -1, p(x', y')
\]
the linear function $\rho(x, y) = x$ is a ranking function, as proved by
the following \emph{SICStus Prolog} session.
\begin{verbatim}
?- use_module(library(clpq)).
% library(clpq) compiled
true.
?- {X >= 0, Y1 =< Y - 1, X1 =< X + Y, Y =< -1,
    X < 1 + X1}.
false.
?- {X >= 0, Y1 =< Y - 1, X1 =< X + Y, Y =< -1,
    X < 0}.
false.
?-
\end{verbatim}
\end{example}

\subsection{Detection}
\label{sec:lrf-detection}

Given an SLC loop, we would like to know whether it admits a linear ranking
function $\rho$.
This problem, which has been studied in depth
\cite{BagnaraMPZ12IC,PodelskiR04,SohnVG91},
is decidable in polynomial time.

Let us consider Example~\ref{ex:linear-ranking-function}
and formally ask whether there exists a ranking function of the form
$\rho(x, y) = a x + b y$:
\begin{multline}
\label{formule-fnrglin}
  \exists a, b \st \forall x, y, x', y' \itc
    \left\{
      \begin{aligned}
        x  &\ge 0,  & x' &\le x+y, \\
        y  &\le -1, & y' &\le y-1
      \end{aligned}
    \right\} \\
  \implies
    \left\{
      \begin{aligned}
        a x + b y &\ge 1 +  a x' + b y', \\
        a x + b y &\ge 0.
      \end{aligned}
    \right.
\end{multline}
This formulation of the problem is executable by quantifier
elimination on a symbolic computation system like
Reduce~\cite{Hearn05}:
\begin{verbatim}
1: load_package redlog;
2: rlset r;
3: F:=ex({a,b},all({x,y,x1,y1},
  (x>=0 and y1<=y-1 and x1<=x+y and y<= -1)
  impl
  (a*x+b*y>=1+a*x1+b*y1 and a*x+b*y>=0)));
4: rlqe F;
\end{verbatim}
Statement~\texttt{1} loads the quantifier elimination module.
Statement~\texttt{2} defines $\Rset$ as the domain of discourse.
Statement~\texttt{3} initializes formula $F$.
Statement~\texttt{4} runs quantifier elimination over $F$ and returns
an equivalent formula, \texttt{true} in this case.
Hence, formula $F$ is true and there exists at least one linear
ranking function.
We can now determine the coefficients of function $\rho$ as follows:
\begin{verbatim}
5: G:=all({x,y,x1,y1},
  (x>=0 and y1<=y-1 and x1<=x+y and y<= -1)
  impl
  (a*x+b*y>=1+a*x1+b*y1 and a*x+b*y>=0));
6: rlqe G;
\end{verbatim}
We obtain
\begin{multline*}
  a^2 - ab \geq 0 \land a - b \neq 0 \land a > 0  \land b = 0 \\
  \land (a^2b - ab^2 \leq 0 \lor a^2 -ab = 0
         \lor a^2 - 2ab - a + b^2 + b \geq 0) \\
  \land (a^2 - ab = 0 \lor a^2 -2ab - a + b^2 b \geq 0),
\end{multline*}
and all values for $a$ and $b$ satisfying the above formula,
such as $a = 1$ and $b = 0$, are equally good.
Unfortunately, the complexity of the algorithms involved
will prevent us from systematically obtaining such a result
within acceptable time and memory bounds.

We now recall the most famous algorithm for this
problem~\cite{PodelskiR04}.\footnote{See also \cite{BagnaraMPZ12IC}.}
Considering $a$ and $b$ as \emph{parameters}
of the problem, we can apply Farkas' Lemma.
For the strict decrease of the ranking function we have
\begin{multline}
\label{decrease-of-the-ranking-function}
  \forall x, y, x', y' \itc
    \left\{
      \begin{aligned}
        x  &\ge 0,  & x' &\le x+y, \\
        y  &\le -1, & y' &\le y-1
      \end{aligned}
    \right\} \\
  \implies
    ax + by \ge 1 + ax' + by'.
\end{multline}
Application of Farkas' Lemma to this problem can be depicted
as follows:
\[
  \begin{array}{lc@{\hspace{1ex}}c@{\;}c@{\hspace{1ex}}c@{\;}c@{\hspace{1ex}}c@{\;}c@{\hspace{1ex}}c@{\;}c@{\hspace{1ex}}c@{\;}c@{\hspace{1ex}}}
    \lambda_1: & 1x & + & 0y & + & 0x' & + & 0y' & + & 0 & \ge & 0 \\
    \lambda_2: & 1x & + & 1y & - & 1x' & + & 0y' & + & 0 & \ge & 0 \\
    \lambda_3: & 0x & + & 1y & + & 0x' & - & 1y' & - & 1 & \ge & 0 \\
    \lambda_4: & 0x & - & 1y & + & 0x' & + & 0y' & - & 1 & \ge & 0 \\
    \implies \\
               & ax & + & by & - & ax' & - & by' & - & 1 & \ge & 0
  \end{array}
\]
We know that formula~(\ref{decrease-of-the-ranking-function})
is equivalent to the existence of
four non-negative rational numbers
$\lambda_1$, \dots,~$\lambda_4$ such that:
\begin{equation}
\label{fnlindec}
  \left\{
    \begin{aligned}
         a &=   \lambda_1 + \lambda_2,
      & -a &=   -\lambda_2, \\
         b &=   \lambda_2 + \lambda_3 - \lambda_4,
      & -b &=   -\lambda_3,
      & -1 &\ge -\lambda_3 -\lambda_4.
    \end{aligned}
  \right.
\end{equation}

The positivity of the ranking function, that is,
\begin{equation}
\label{positivity-of-the-ranking-function}
  \forall x, y, x', y' \itc
    \left\{
      \begin{aligned}
        x  &\ge 0  & x' &\le x+y \\
        y  &\le -1 & y' &\le y-1 \\
      \end{aligned}
    \right\}
  \implies
    ax + by \ge 0
\end{equation}
can be written as
\[
  \begin{array}{lc@{\hspace{1ex}}c@{\;}c@{\hspace{1ex}}c@{\;}c@{\hspace{1ex}}c@{\;}c@{\hspace{1ex}}c@{\;}c@{\hspace{1ex}}c@{\;}c@{\hspace{1ex}}}
    \lambda'_1: & 1x & + & 0y & + & 0x' & + & 0y' & + & 0 & \ge & 0 \\
    \lambda'_2: & 1x & + & 1y & - & 1x' & + & 0y' & + & 0 & \ge & 0 \\
    \lambda'_3: & 0x & + & 1y & + & 0x' & - & 1y' & - & 1 & \ge & 0 \\
    \lambda'_4: & 0x & - & 1y & + & 0x' & + & 0y' & - & 1 & \ge & 0 \\
    \implies \\
                & ax & + & by & + & 0x' & + & 0y' & + & 0 & \ge & 0.
\end{array} \]
By Farkas' Lemma, formula~\eqref{positivity-of-the-ranking-function}
is equivalent to the existence of four other non-negative rational numbers
$\lambda'_1$, \dots,~$\lambda'_4$ such that:
\begin{equation}
\label{fnlinpos}
  \left\{
    \begin{aligned}
      a &= \lambda'_1 + \lambda'_2,
                                         & 0 &=   -\lambda'_2, \\
      b &= \lambda'_2 + \lambda'_3 - \lambda'_4,
                                         & 0 &=   -\lambda'_3, &
        &
                                         & 0 &\ge -\lambda'_3 - \lambda'_4.
    \end{aligned}
  \right.
\end{equation}

Summarizing, by Farkas Lemma,
formula~\eqref{formule-fnrglin} is equivalent to the conjunction
of formulas~\eqref{fnlindec} and \eqref{fnlinpos}:
\begin{multline}
\label{fnlinrf}
  \exists a, b  \st
    \exists \lambda_1, \ldots, \lambda_4,
            \lambda'_1, \ldots, \lambda'_4 \ge 0 \st \\
    \left\{
      \begin{aligned}
         a  &= \lambda_1 + \lambda_2,
                                & -a &= -\lambda_2, \\
         b  &= \lambda_2 + \lambda_3 - \lambda_4,
                                & -b &= -\lambda_3, \\
         a  &=   \lambda'_1+\lambda'_2,
                                &  0 &= -\lambda'_2, \\
         b  &=   \lambda'_2+\lambda'_3-\lambda'_4,
                                &  0 &= -\lambda'_3, \\
         -1 &\ge -\lambda_3 -\lambda_4,
                                &  0  &\ge -\lambda'_3 -\lambda'_4.
      \end{aligned}
    \right.
\end{multline}

In theory, the problem of the existence of a linear ranking function
is polynomial. Since computing one solution (that is, values for
$a$ and $b$) is not harder than determining its existence,
a ``witness'' function, which would constitute a
\emph{termination certificate}, can also be computed in polynomial time.

The space of all linear ranking functions as defined in Definition \ref{def-fn-rng-lin},
described by parameters $a$ and $b$,
can be obtained by elimination of $\lambda_i$ and $\lambda'_i$
from~\eqref{fnlinrf} using, e.g., the algorithm of Fourier-Motzkin.
For example the SICStus Prolog program
\begin{verbatim}
fm(A, B) :-
   {L1 >= 0, L2 >= 0, L3 >= 0, L4 >= 0,
    LP1 >= 0, LP2 >= 0, LP3 >= 0, LP4 >= 0,
    A = L1 + L2, B = L2 + L3 - L4,
    A = L2, B = L3, 1 =< L3 + L4,
    A = LP1 + LP2, B = LP2 + LP3 - LP4,
    0 = LP2, 0 = LP3, 0 =< LP3 + LP4}.
\end{verbatim}
can be queried as follows:
\begin{verbatim}
| ?- fm(A, B).
B = 0, {A >= 1}.
| ?-
\end{verbatim}
It can be shown that the computed answer is equivalent to the
(significantly more involved) condition generated by Reduce.

\section{Eventual Linear Ranking Functions}
\label{sec:eventual-linear-ranking-functions}

In the previous section we have illustrated a method to decide
the existence of a linear ranking function for a rational SLC loop,
something that implies termination of the loop.
Of course, the method cannot decide termination in all cases.

\begin{example}
\label{ex-fn-rng-lin-evt-p}
The loop
\[
  p(x, y) \leftarrow x \ge 0, y' \le y - 1, x' \le x + y, p(x', y')
\]
does not admit a linear ranking function.
\end{example}

Can we conclude that such loop does not always terminate?
No, because it may admit a non-linear ranking function.

In this section we will extend the previous method so as to detect
\emph{eventual linear ranking functions}, that is, linear functions
that behave as ranking functions
\emph{after a finite number of executions of the loop body}.
Suppose that the considered SLC loop is always given with
a linear function $f(x, y)$ that increases at each
iteration of the loop in the following sense:

\begin{definition}
\label{def-inc-lin-fn}
Let $C$ be the SLC loop
$p(\mathbf{x}) \leftarrow c(\mathbf{x}, \mathbf{x}'),$ $p(\mathbf{x}')$.
A function $f(\mathbf{x})$ is \emph{increasing for $C$}
if it is linear and satisfies:
$
  \forall \mathbf{x}, \mathbf{x}'
    \itc
      c(\mathbf{x}, \mathbf{x}')
        \implies f(\mathbf{x}') \ge 1 + f(\mathbf{x}).
$
\end{definition}

\begin{example}
\label{ex-fn-rng-lin-evt-f}
The function $f(x, y) = -y$ is increasing for the loop of
\textup{Example~\ref{ex-fn-rng-lin-evt-p}}, since $y$ decreases by at
least $1$ at each iteration.
\end{example}

\begin{remark}
The generalization to affine functions  is useless.
Moreover, as we are merely interested in the existence of an
increasing function, the value of the increase ($1$ or $\varepsilon > 0$)
is irrelevant.
\end{remark}

We can now give the definition which is central to our paper.

\begin{definition}
\label{def-fn-rng-lin-evt}
Let $C$ be the rational SLC loop in clausal form
$p(\mathbf{x}) \leftarrow c(\mathbf{x}, \mathbf{x}'), \; p(\mathbf{x}')$,
where $p$ is an $n$-ary relation;
let also $f(\mathbf{x})$ be a linear increasing function for $C$.
An \emph{eventual linear ranking function} $\rho$ for $(C, f)$
is a linear map of $\Qset^n$ to $\Qset$ such that
\begin{multline*}
  \exists k \st
    \forall \mathbf{x}, \mathbf{x}' \itc
      \bigl(
        c(\mathbf{x}, \mathbf{x}')
        \,\land\,
        f(\mathbf{x}) \ge k
      \bigr) \\
        \implies
          \bigl(
            \rho(\mathbf{x}) \ge 1 + \rho(\mathbf{x}')
            \,\land\,
            \rho(\mathbf{x}) \ge 0
          \bigr).
\end{multline*}
\end{definition}
For comparison with Definition~\ref{def-fn-rng-lin},
remark that the threshold $k$ is existentially quantified
and that $f(\mathbf{x}) \ge k$ is imposed in the implication antecedent.
It should also be noted that, if such a rational $k$ exists,
then each $k' \ge k$ satisfies the condition
of Definition~\ref{def-fn-rng-lin-evt}.
On the other hand, since, by hypothesis, $f$ strictly increases
at each iteration, there are two cases:
either $f$ is bounded from above by a constant, and thus the loop will
terminate;
or, after a finite number of iterations, $f$ will cross the threshold $k$
and $\rho$ becomes a linear ranking function in the sense of
Section~\ref{sec:linear-ranking-functions} so that, again, the loop terminates.

Eventual linear ranking functions are a generalization
of linear ranking functions.

\begin{proposition}
\label{lrf-implies-elrf}
Let $C$ be an SLC loop. If $\rho$ is  a linear ranking function for $C$,
then there exists an increasing function $f$ such that $(C, f)$
has an eventual linear ranking function.
\end{proposition}
\begin{proof}
By hypothesis, there exists a linear ranking function
$\rho(\mathbf{x})$ for $C$.
The linear function $f(\mathbf{x}) \overset{\mathtt{def}}{=} -\rho(\mathbf{x})$
is non-positive and strictly increasing for $C$.
Considering $k = 1$ it can be seen that the function
$\rho'(\mathbf{x})\overset{\mathtt{def}}{=}0$ is an eventual linear
ranking function for $(C, f)$.
\end{proof}

The generalization is strict as the loop of
Example~\ref{ex-fn-rng-lin-evt-p} has no linear ranking function, but
does have an eventual linear ranking function, as will be shown in the
next section.

\subsection{Detection given a Linear Increasing Function}
\label{sec:elrf-semi-detection}

As a first step towards full automation of the synthesis of eventual
linear ranking functions, we assume that an SLC loop is given with a
particular linear increasing function.
Let us consider, e.g., the SLC loop of
Example~\ref{ex-fn-rng-lin-evt-p} and the increasing
function of Example~\ref{ex-fn-rng-lin-evt-f}.
Defining $\rho(x, y) = a x + b y$,
$\rho$~is an eventual linear ranking function when
\begin{multline*}
  \exists a, b, k  \st \forall x, y, x', y' \itc
    \left\{
      \begin{aligned}
         x &\ge 0, & x' &\le x + y, \\
        -y &\ge k, & y' &\le y - 1
       \end{aligned}
    \right\} \\
      \implies
        \left\{
          \begin{aligned}
            ax + by &\ge 1 + ax' + by', \\
            ax + by &\ge 0.
          \end{aligned}
        \right.
\end{multline*}
This definition of the problem, that we will denote for brevity with
$\exists a, b, k \st \phi(a, b, k)$, is also solvable via quantifier
elimination, hence the problem is decidable.
Considering $a$, $b$ and $k$ as parameters,
we can apply Farkas' Lemma as follows:
\[
  \begin{array}{lc@{\hspace{1ex}}c@{\;}c@{\hspace{1ex}}c@{\;}c@{\hspace{1ex}}c@{\;}c@{\hspace{1ex}}c@{\;}c@{\hspace{1ex}}c@{\;}c@{\hspace{1ex}}}
    \lambda_1: & 1x & + & 0y & + & 0x' & + & 0y' & + & 0 & \ge & 0 \\
    \lambda_2: & 1x & + & 1y & - & 1x' & + & 0y' & + & 0 & \ge & 0 \\
    \lambda_3: & 0x & + & 1y & + & 0x' & - & 1y' & - & 1 & \ge & 0 \\
    \lambda_4: & 0x & - & 1y & + & 0x' & + & 0y' & - & k & \ge & 0 \\
    \implies \\
               & ax & + & by & - & ax' & - & by' & - & 1 & \ge & 0 \\
               & ax & + & by &   &     &   &     &   &   & \ge & 0.
  \end{array}
\]
Hence, formula $\phi(a, b, k)$ is equivalent to the conjunction
of formulas $\dec(a, b, k)$, i.e.,
\begin{align*}
  \exists &\lambda_1 \ge 0, \ldots, \lambda_4 \ge 0
  \st \\
    &\left\{
      \begin{aligned}
         a &=    \lambda_1 + \lambda_2,
                                         & -a &=   -\lambda_2, \\
         b &=    \lambda_2 + \lambda_3 - \lambda_4,
                                         & -b &=   -\lambda_3, &
        -1 &\ge -\lambda_3 - k \lambda_4,
      \end{aligned}
    \right.
\intertext{%
ensuring the decreasing of the ranking function,
and the formula $\pos(a, b, k)$, that is,
}
  \exists &\lambda'_1 \ge 0, \ldots, \lambda'_4 \ge 0
  \st \\
    &\left\{
      \begin{aligned}
        a &=    \lambda'_1 + \lambda'_2,
                                         & 0 &=   -\lambda'_2 \\
        b &=    \lambda'_2 + \lambda'_3 - \lambda'_4,
                                         & 0 &=   -\lambda'_3 &
        0 &\ge -\lambda'_3 - k\lambda'_4,
      \end{aligned}
    \right.
\end{align*}
ensuring the positivity of the ranking function.

Let us focus on $\dec(a, b, k)$.
We observe that the product $k \lambda_4$
leads to a non-linearity that we can circumvent
by noting that, as $\lambda_4 \ge 0$,
either $\lambda_4 = 0$ (hence $k\lambda_4 = 0$) or $\lambda_4 > 0$.
In the latter case, we introduce a new variable $P= k \lambda_4$.
We have the property:
\begin{lemma}
\label{lem1-fn-rng-lin-evt}
Formula $\exists k \st \dec(a, b, k)$ is equivalent to the disjunction
$\dec_1(a, b) \lor \dec_2(a, b)$.
\end{lemma}
In our case, $\dec_1(a, b)$ is equivalent to
\begin{align*}
  \exists &\lambda_1, \lambda_2, \lambda_3 \ge 0
    \st \\
      &\left\{
        \begin{aligned}
          a  &=   \lambda_1 + \lambda_2,  & -a &=   -\lambda_2, \\
          b  &=    \lambda_2 + \lambda_3, & -b &=   -\lambda_3, &
          -1 &\ge -\lambda_3,
        \end{aligned}
      \right.
\intertext{%
and $\dec_2(a, b)$ is equivalent to
}
  \exists &\lambda_1 \ge 0,  \lambda_2 \ge 0,  \lambda_3 \ge 0, \lambda_4 > 0, P
    \st \\
      &\left\{
        \begin{aligned}
            a &=    \lambda_1+\lambda_2,           & -a &=   -\lambda_2, \\
            b &=    \lambda_2+\lambda_3-\lambda_4, & -b &=   -\lambda_3, &
           -1 &\ge -\lambda_3 - P.
        \end{aligned}
      \right.
\end{align*}

\begin{proof}
($\Longrightarrow$) Let  $k$ be a rational number and
$\lambda_i$'s for  $1 \le i \le 4$ four non-negative rational numbers
such that $\dec(a, b, k)$ holds.
If $\lambda_4 = 0$ then  $\dec(a, b, k)$ simplifies to $\dec_1(a, b)$
which is true.
If $\lambda_4 > 0$, we take $P = k\lambda_4$ and
we can see that $\dec_2(a, b)$ is true.

\noindent
($\Longleftarrow$) Assume first that $\dec_1(a, b)$ is true.
Then, taking $\lambda_4 = 0$ and $k = 0$ (any rational number
would be fine for $k$), we see that $\exists k \st \dec(a, b, k)$ is true.
Assume then that $\dec_2(a, b)$  is true.
Taking $k = P/\lambda_4$ (this is always possible as
$\lambda_4 >0$), we observe that there exists $k$
such that  $\dec(a, b, k)$ is true.
\end{proof}

For the positivity condition, we can prove in a similar way
\begin{lemma}
\label{lem2-fn-rng-lin-evt}
Formula $\exists k \st \pos(a, b, k)$ is equivalent to the disjunction
$\pos_1(a, b) \lor \pos_2(a, b)$.
\end{lemma}
In our case, $\pos_1(a, b)$ is equivalent to
\begin{align*}
  \exists &\lambda'_1, \lambda'_2, \lambda'_3 \ge 0
    \st \\
      &\left\{
        \begin{aligned}
          a &=    \lambda'_1 + \lambda'_2, & 0 &=   -\lambda'_2, \\
          b &=    \lambda'_2 + \lambda'_3, & 0 &=   -\lambda'_3, &
          0 &\ge -\lambda'_3,
        \end{aligned}
      \right.
\intertext{%
and $\pos_2(a, b)$ to
}
  \exists &\lambda'_1, \lambda'_2, \lambda'_3 \ge 0, \lambda'_4 > 0,
          P'
    \st \\
      &\left\{
        \begin{aligned}
          a &=    \lambda'_1 + \lambda'_2,              & 0 &=   -\lambda'_2, \\
          b &=    \lambda'_2 + \lambda'_3 - \lambda'_4, & 0 &=   -\lambda'_3, &
          0 &\ge -\lambda'_3  -P'.
        \end{aligned}
      \right.
\end{align*}

Combining the previous results gives
\begin{proposition}
Formula $\exists k.\phi(a, b, k)$ is equivalent to
$[\dec_1(a, b) \lor \dec_2(a, b)] \land [\pos_1(a, b) \lor \pos_2(a, b)]$.
\end{proposition}

\begin{proof}
Thanks to the previous lemmata, it only remains to justify the equivalence
between the formulas $\exists k \st \phi(a, b, k)$
and $\exists k \st \dec(a, b, k) \land \exists k \st \pos(a, b, k)$.

($\Longrightarrow$) Let $k_0$ be a rational such that $\phi(a, b, k_0)$.
We have $\dec(a, b, k_0)$ and $\pos(a, b, k_0)$ because
\[
  \phi(a, b, k) \iff \dec(a, b, k) \land \pos(a, b, k).
\]

($\Longleftarrow$) Assume the existence of $k_d$
such that  $\dec(a, b, k_d)$ and the existence of $k_p$
such that $\pos(a, b, k_p)$.
Then the rational $k_0 = \max(k_d, k_p)$
verifies $\dec(a, b, k_0) \land \pos(a, b, k_0)$
and shows that $\exists k \st \phi(a, b, k)$.
\end{proof}

Back to our initial problem, the existence of an eventual linear
ranking function is equivalent to the satisfiability of at least one
of the following four linear systems:
\begin{align*}
  \dec_1(a, b) &\land \pos_1(a, b), \\
  \dec_1(a, b) &\land \pos_2(a, b), \\
  \dec_2(a, b) &\land \pos_1(a, b), \\
  \dec_2(a, b) &\land \pos_2(a, b),
\end{align*}
which we can decide in polynomial time.
For our running example,
$\dec_2(a, b) \land \pos_1(a, b)$ is satisfiable
as proved by the following SICStus Prolog query:
\begin{verbatim}
?- dec2pos1.
true.
?-
\end{verbatim}
after compilation of the program:
\begin{verbatim}
dec2pos1 :-
   {L1 >= 0, L2 >= 0, L3 >= 0, L4 > 0,
    A = L1 + L2, B = L2 + L3 - L4,
    A = L2, B = L3, 1 =< L3 + P,
    LP1 >= 0, LP2 >= 0, LP3 >= 0,
    A = LP1 + LP2, B = LP2 + LP3,
    0 = LP2, 0 = LP3, 0 =< LP3}.
\end{verbatim}

The procedure we have informally outlined by means of examples is
actually completely general.
It is embodied in Algorithm~\ref{algo1}, which is a (correct and
complete) decision procedure for the existence of an eventual
linear ranking function given a linear increasing function.

\begin{algorithm}
\caption{Existence of an eventual linear ranking function, given a linear increasing function}
\label{algo1}
\begin{algorithmic}[1]
\REQUIRE
$C$, an SLC loop
$p(\mathbf{x}) \leftarrow c(\mathbf{x}, \mathbf{x}'), p(\mathbf{x}')$,
and $f$, a linear increasing function for $C$
\ENSURE
Returns \TRUE\ if and only if, for some vector $\mathbf{a}$,
$\rho(\mathbf{x}) = \mathbf{a} \mathbf{x} =  \sum_i a_i x_i$
is an eventual linear ranking function for $(C, f)$.
\STATE
$\dec(\mathbf{a}, k) \gets \text{Farkas for the decreasing of $\rho$}$
\STATE
$\dec_1(\mathbf{a}), \dec_2(\mathbf{a}) \gets \text{linearization of  $\dec(\mathbf{a}, k)$}$
\STATE
$\pos(\mathbf{a}, k) \gets \text{Farkas for the positivity of $\rho$}$
\STATE
$\pos_1(\mathbf{a}), \pos_2(\mathbf{a}) \gets \text{linearization of $\pos(\mathbf{a}, k)$}$
\IF {$\bigvee_{1\le i, j \le2} \dec_i(\mathbf{a}) \land \pos_j(\mathbf{a})$ is satisfiable}
  \RETURN \TRUE
\ELSE
  \RETURN \FALSE
\ENDIF
\end{algorithmic}
\end{algorithm}

\begin{theorem}
\label{algo-is-a-decision-procedure-for-elrf}
Let $C$ be an SLC loop and $f$ an increasing function for $C$.
\textup{Algorithm~\ref{algo1}} decides in polynomial time the existence
of an eventual linear ranking function for $(C, f)$.
\end{theorem}

Computing an eventual linear ranking function $\rho$
and its associated threshold $k$ can be done as follows:
\begin{itemize}
\item
if $\dec_1(\mathbf{a}) \land \pos_1(\mathbf{a})$ is satisfiable,
we compute a solution $\mathbf{a}$,
$\rho(\mathbf{x}) = \mathbf{a}\mathbf{x}$ is a standard linear ranking function
and Proposition~\ref{lrf-implies-elrf} applies;
\item
if $\dec_1(\mathbf{a}) \land \pos_2(\mathbf{a})$ is satisfiable,
we compute a solution $\mathbf{a}$, $\mathbf{\lambda}'$, $P'$
and we take $k = P'/\lambda'_n$;
\item
if $\dec_2(\mathbf{a}) \land \pos_1(\mathbf{a})$ is satisfiable,
we compute a solution $\mathbf{a}$, $\mathbf{\lambda}$, $P$
and we take $k = P/\lambda_n$;
\item
if $\dec_2(\mathbf{a}) \land \pos_2(\mathbf{a})$ is satisfiable,
we compute a solution
$\mathbf{a}$, $\mathbf{\lambda}$, $P$, $\mathbf{\lambda}'$, $P'$
and we take $k = \max(P/\lambda_n, P'/\lambda'_n)$.
\end{itemize}

\begin{example}
Continuing with \textup{Example~\ref{ex-fn-rng-lin-evt-p}},
here is the most general solution of $\dec_2(a, b) \land \pos_1(a, b)$:
\begin{verbatim}
?-  {L1 >= 0, L2>= 0, L3 >= 0, L4 > 0,
     A = L1 + L2, B = L2 + L3 - L4,
     A = L2, B = L3, 1 =< L3 + P,
     LP1 >= 0, LP2 >= 0, LP3 >= 0,
     A = LP1 + LP2, B = LP2 + LP3,
     0 = LP2, 0 = LP3, 0 =< LP3}.
B = 0, L1 = 0, L3 = 0, LP2 = 0, LP3 = 0,
{LP1 = L4, L2 = L4, A = L4, L4 > 0, P >= 1}.
?-
\end{verbatim}
One particular solution is
$b = 0 = \lambda_1  = \lambda_3 = \lambda'_2 = \lambda'_3$,
$a = 1 = \lambda'_1 = \lambda_2 = \lambda_4$,
$P = 1$.
Hence $\rho(x, y) = x$ is an eventual linear ranking function
from the threshold $k = P/\lambda_4 = 1$.
\end{example}

We also provide a decision procedure for the existence
of an eventual \emph{affine} ranking function.

\begin{corollary}
The existence of an eventual affine ranking function for
an SLC loop and associated increasing function, $(C, f)$,
can be decided in polynomial time.
\end{corollary}
\begin{proof}
From $C$,
$p(\mathbf{x}) \leftarrow c(\mathbf{x},\mathbf{x}'), p(\mathbf{x}')$,
we construct $C_a$,
\(
  p(\mathbf{x},y)
    \leftarrow c(\mathbf{x},\mathbf{x}'), \; y = 1 = y', \; p(\mathbf{x}',y'),
\) where $y$ does not occur in $\mathbf{x}$.
Note that $C_a$ is an SLC loop and that $f_a(\mathbf{x},y) = f(\mathbf{x})$
is an increasing function for $C_a$.
Algorithm~\ref{algo1} applied to $(C_a, f_a)$
gives an answer in polynomial time.

If Algorithm~\ref{algo1} returns \textbf{true} then, by correctness,
there exists a threshold $k$ and an eventual linear function
$\rho_a(\mathbf{x},y) = \mathbf{a} \mathbf{x} + b y$ for $(C_a, f_a)$.
We readily check that $\rho(\mathbf{x}) = \mathbf{a}\mathbf{x} + b$
is an eventual affine ranking function for $(C, f)$ from $k$.

If Algorithm~\ref{algo1} returns \textbf{false} then, by completeness,
there is no eventual linear ranking function for $(C_a, f_a)$.
Assuming there exists an eventual affine ranking function
$\rho(\mathbf{x}) = \mathbf{a} \mathbf{x} + b$ from $k$ for $(C, f)$, then
$\rho_a(\mathbf{x},y) = \mathbf{a} \mathbf{x} + b y$ should be an
eventual linear ranking function from $k$ for $(C_a,f_a)$,
which is a contradiction.
Hence there is no eventual affine ranking function for $(C, f)$.
\end{proof}

\begin{example}
\label{ex-fn-rng-aff-evt-p}
The SLC loop
\[
  p(x, y) \leftarrow x \ge -1, y' \le y - 1, x' \le x + y, p(x', y')
\]
associated to the linear increasing function $f(x,y) = -y$
does not admit an eventual linear ranking function, but does admit
$\rho(x,y)=x+1$ as an eventual affine ranking function from $k=1$.
\end{example}

\subsection{Fully Automated Detection}
\label{sec:elrf-full-detection}

We now consider the problem in its full generality:
given an SLC loop $C$, does there exist an increasing function for $C$
such that $C$ admits an eventual linear ranking function?

Note that the space of increasing functions can be obtained
as a convex set over their coefficients
via the Farkas' Lemma and existentially quantified variables
elimination.\footnote{See also \cite[Section 4.4]{BagnaraMPZ12IC}.}

\begin{definition}
\label{def-inc}
Let
\(
  C = \bigl(
        p(\mathbf{x}) \leftarrow c(\mathbf{x}, \mathbf{x}'), p(\mathbf{x}')
      \bigr)
\)
be an SLC loop.
We denote by $\inc$ the set of vectors $\mathbf{b}$
such that $f(\mathbf{x}) = \mathbf{b}\mathbf{x} = \sum_i b_i x_i$
is increasing for $C$.
\end{definition}

\begin{example}
\label{ex-full-detection}
A linear ranking function does not exist for the SLC loop $C$
\[
  p(x, y) \leftarrow x \ge 0, x' \le x+y, y' \le  -y - 1, p(x', y').
\]
\(
  \inc = \bigl\{\,
                   (b_1,b_2) \in \Qset \times \Qset
                  \bigm|
                    b_1 \le -2, b_1 - 2b_2 = 0
                  \,\bigr\}
\)
induces the space of functions of the form $f(x, y)= b_1 x + b_2 y$,
which are increasing for $C$.
\end{example}

Let us consider the SLC loop of Example~\ref{ex-full-detection}
associated to an increasing function $f(x, y) = b_1 x + b_2 y$
induced by $\inc$.
Defining $\rho(x, y) = a_1 x + a_2 y$ and
considering $b_1$ and $b_2$ as parameters,
$\rho$ is an eventual linear ranking function when
\begin{multline*}
  \exists a_1, a_2, k  \st \forall x, y, x', y' \itc
    \left\{
      \begin{aligned}
         x &\ge 0, & x' \le x + y, \\
         b_1 x + b_2 y &\ge k, & y' \le -y - 1
       \end{aligned}
    \right\} \\
      \implies
        \left\{
          \begin{aligned}
            a_1 x + a_2 y &\ge 1 + a_1 x' + a_2 y', \\
            a_1 x + a_2 y &\ge 0.
          \end{aligned}
        \right.
\end{multline*}
This definition of the problem is denoted $\exists \mathbf{a}, k \st \phi(\mathbf{a}, k)$.
We can apply Farkas' Lemma as follows:
\[
  \begin{array}{lc@{\hspace{1ex}}c@{\;}c@{\hspace{1ex}}c@{\;}c@{\hspace{1ex}}c@{\;}c@{\hspace{1ex}}c@{\;}c@{\hspace{1ex}}c@{\;}c@{\hspace{1ex}}}
    \lambda_1: & 1x & + & 0y & + & 0x' & + & 0y' & + & 0 & \ge & 0 \\
    \lambda_2: & 1x & + & 1y & - & 1x' & + & 0y' & + & 0 & \ge & 0 \\
    \lambda_3: & 0x & - & 1y & + & 0x' & - & 1y' & - & 1 & \ge & 0 \\
    \lambda \, \, \, : & b_1 x & + & b_2 y & + & 0x' & + & 0y' & - & k & \ge & 0 \\
    \implies \\
               & a_1 x & + & a_2y & - & a_1 x' & - & a_2 y' & - & 1 & \ge & 0 \\
               & a_1 x & + & a_2y &   &     &   &     &   &   & \ge & 0.
  \end{array}
\]
Formula $\phi(\mathbf{a}, k)$ is equivalent to the conjunction
of formulas $\dec(\mathbf{a}, k)$, i.e.,
\begin{align*}
  \exists &\lambda_1 \ge 0, \lambda_2 \ge 0, \lambda_3\ge 0, \lambda \ge 0
  \st \\
    &\left\{
      \begin{aligned}
         a_1 &=    \lambda_1 + \lambda_2 + b_1 \lambda
                                         & -a_1 &=   -\lambda_2, \\
         a_2 &=    \lambda_2 - \lambda_3 + b_2\lambda,
                                         & -a_2 &=   -\lambda_3, &
        -1 &\ge -\lambda_3 - k \lambda,
      \end{aligned}
    \right.
\intertext{%
ensuring the decreasing of the ranking function
and $\pos(\mathbf{a},k)$, that is,
}
  \exists &\lambda'_1 \ge 0, \lambda'_2 \ge 0, \lambda'_3\ge 0, \lambda' \ge 0
  \st \\
    &\left\{
      \begin{aligned}
        a_1 &=    \lambda'_1 + \lambda'_2 +b_1 \lambda'
                                         & 0 &=   -\lambda'_2 \\
        a_2 &=    \lambda'_2 - \lambda'_3 + b_2 \lambda',
                                         & 0 &=   -\lambda'_3 &
        0 &\ge -\lambda'_3 - k\lambda',
      \end{aligned}
    \right.
\end{align*}
ensuring the positivity of the ranking function.

Let us focus on $\dec(\mathbf{a}, k)$.
We observe that the products with $\lambda$
lead to a non-linearity that we can circumvent
by noting that, as $\lambda \ge 0$,
either $\lambda = 0$ or $\lambda > 0$.
In the latter case, we introduce a vector $\mathbf{p}=(p_1,p_2)$ of
two new variables where
$p_1 = b_1 \lambda$ and $p_2 = b_2 \lambda$ together
with, as previously, the new variable $P= k \lambda$.
Formula $\exists k \st \dec(\mathbf{a}, k)$ is equivalent to the disjunction
$\dec_1(\mathbf{a}) \lor \exists \lambda, \mathbf{p}. \dec_2(\mathbf{a},\lambda,
\mathbf{p})$ where
in our case, $\dec_1(\mathbf{a})$ is equivalent to
\begin{align*}
  \exists &\lambda_1 \ge 0,  \lambda_2 \ge 0,  \lambda_3 \ge 0
    \st \\
      &\left\{
        \begin{aligned}
          a_1  &=   \lambda_1 + \lambda_2,  & -a_1 &=   -\lambda_2, \\
          a_2  &=    \lambda_2 - \lambda_3, & -a_2 &=   -\lambda_3, &
          -1 &\ge -\lambda_3,
        \end{aligned}
      \right.
\intertext{%
and $\dec_2(\mathbf{a},\lambda, \mathbf{p})$ is equivalent to
}
  \exists &\lambda_1 \ge 0,  \lambda_2 \ge 0,  \lambda_3 \ge 0, P
    \st \\
      &\left\{
        \begin{aligned}
            a_1 &=    \lambda_1+\lambda_2+p_1, & -a_1 &=   -\lambda_2,  &\lambda &> 0,\\
            a_2 &=    \lambda_2-\lambda_3+p_2, & -a_2 &=   -\lambda_3, &
           -1 &\ge -\lambda_3 - P.
        \end{aligned}
      \right.
\end{align*}

For the positivity condition,
formula $\exists k \st \pos(\mathbf{a}, k)$ is equivalent to the disjunction
$\pos_1(\mathbf{a}) \lor \exists \lambda', \mathbf{p'}. \pos_2(\mathbf{a},\mathbf{p'})$
where we introduce a vector $\mathbf{p'}=(p'_1,p'_2)$ of
two new variables where
$p'_1=b_1 \lambda'$, $p'_2=b_2 \lambda'$ together
with, as previously, the new variable $P'= k \lambda'$.
In our case, $\pos_1(\mathbf{a})$ is equivalent to
\begin{align*}
  \exists &\lambda'_1 \ge 0,  \lambda'_2 \ge 0,  \lambda'_3 \ge  0
    \st \\
      &\left\{
        \begin{aligned}
          a_1 &=    \lambda'_1 + \lambda'_2, & 0 &=   -\lambda'_2, \\
          a_2 &=    \lambda'_2 - \lambda'_3, & 0 &=   -\lambda'_3, &
          0 &\ge -\lambda'_3,
        \end{aligned}
      \right.
\intertext{%
and $\pos_2(\mathbf{a},\lambda',\mathbf{p'})$ to
}
  \exists &\lambda'_1 \ge 0,  \lambda'_2 \ge 0,  \lambda'_3 \ge  0,  P'
    \st \\
      &\left\{
        \begin{aligned}
          a_1 &= \lambda'_1 + \lambda'_2 +p'_1,  & 0 &= -\lambda'_2, & \lambda' &> 0, \\
          a_2 &=    \lambda'_2 - \lambda'_3 +p'_2,  & 0 &=   -\lambda'_3, &
          0 &\ge -\lambda'_3  -P'.
        \end{aligned}
      \right.
\end{align*}

Back to our initial problem, the existence of an eventual linear
ranking function is equivalent to the satisfiability of at least one
of the following four systems:
\begin{enumerate}
\item
$\dec_1(\mathbf{a}) \land \pos_1(\mathbf{a})$:
this case means that the increasing function
and $k$ are irrelevant. In other words, for each solution $\mathbf{a}$,
$\rho(\mathbf{x}) = \mathbf{a}\mathbf{x}$ is a standard linear ranking function
and Proposition~\ref{lrf-implies-elrf} applies.

\item
\(
  \dec_1(\mathbf{a})
    \land \pos_2(\mathbf{a}, \lambda', \mathbf{p'})
    \land \mathbf{p'}/\lambda' \in \inc
\):
note that satisfiability of
$\dec_1(\mathbf{a}) \land \pos_2(\mathbf{a}, \lambda', \mathbf{p'})$
is not sufficient, as its solution might lead to the coefficients
$b_1 = p'_1/\lambda'$ and $b_2=p'_2/\lambda'$
($\lambda'$ is strictly positive by definition),
which could correspond to a non-increasing linear function.
The third conjunct, $\mathbf{p'}/\lambda' \in \inc$,
ensures that we stay within the space of increasing functions.

\item
\(
  \dec_2(\mathbf{a}, \lambda, \mathbf{p})
    \land \mathbf{p}/\lambda \in \inc
    \land \pos_1(\mathbf{a})
\):
this case is symmetric to previous one.

\item
\(
  \dec_2(\mathbf{a},\lambda,\mathbf{p})  \land \mathbf{p}/\lambda \in \inc
    \land \pos_2(\mathbf{a},  \lambda',\mathbf{p'})
    \land \mathbf{p'}/\lambda' \in \inc
    \land \mathbf{p} / \lambda = \mathbf{p'} / \lambda'
\):
this case combines the two previous ones.
Note that the condition ensures that we consider the same linear
ranking function and the same increasing function both in $\dec_2$ and
in $\pos_2$.
\end{enumerate}

For our running example, the following SICStus Prolog query
proves that
\(
  \dec_2(\mathbf{a}, \lambda,\mathbf{p})
    \land \mathbf{p}/\lambda \in \inc
    \land \pos_1(\mathbf{a})
\)
is satisfiable
\begin{verbatim}
?- dec2incpos1.
true.
?-
\end{verbatim}
after compilation of the program
\begin{verbatim}
dec2incpos1 :-
   {% DEC2:
    L1 >= 0, L2 >= 0, L3 >= 0,
    A1 = L1 + L2 + P1, A1 = L2, L > 0,
    A2 = L2 - L3 + P2, A2 = L3, -1 >= -L3 - P,
    % INC: B1 =< -2, B1 - 2*B2 = 0
    P1 =< -2*L, P1 - 2*P2 = 0,
    % POS1:
    LP1 >= 0, LP2 >= 0, LP3 >= 0,
    A1 = LP1 + LP2, 0 = LP2,
    A2 = LP2 - LP3, 0 = LP3, 0 >= -LP3}.
\end{verbatim}

The procedure we have informally outlined by means of examples is
actually completely general and is embodied in Algorithm~\ref{algo2}.

\begin{algorithm}
\caption{Existence of an eventual linear ranking function}
\label{algo2}
\begin{algorithmic}[1]
\REQUIRE
$C$, an SLC loop
$p(\mathbf{x}) \leftarrow c(\mathbf{x}, \mathbf{x}'), p(\mathbf{x}')$
\ENSURE
Returns \TRUE\ if and only if there exists an increasing function $f$
for $C$ and
$\rho(\mathbf{x}) = \mathbf{a} \mathbf{x}$
such that $\rho$ is an eventual linear ranking function for $(C, f)$.
\STATE
$\inc \gets \text{the space of increasing functions for $C$}$
\STATE
$\dec(\mathbf{a}, k) \gets \text{Farkas for the decreasing of $\rho$}$
\STATE
$\dec_1(\mathbf{a}), \dec_2(\mathbf{a}, \lambda, \mathbf{p}) \gets \text{linearization of $\dec(\mathbf{a}, k)$}$
\STATE
$\pos(\mathbf{a}, k) \gets \text{Farkas for the positivity of $\rho$}$
\STATE
$\pos_1(\mathbf{a}), \pos_2(\mathbf{a},\lambda', \mathbf{p'}) \gets \text{linearization of $\pos(\mathbf{a}, k)$}$
\STATE
$\phi_{1,1} \gets \dec_1(\mathbf{a}) \land \pos_1(\mathbf{a})$
\STATE
 $\phi_{1,2} \gets \dec_1(\mathbf{a}) \land \pos_2(\mathbf{a},\lambda',\mathbf{p'})
 					\land \mathbf{p'} / \lambda' \in \inc$
 \STATE
 $\phi_{2,1} \gets \dec_2(\mathbf{a},\lambda,\mathbf{p}) \land \mathbf{p} / \lambda \in \inc
 					\land \pos_1(\mathbf{a})$
 \STATE
$\phi_{2,2} \gets \dec_2(\mathbf{a},\lambda,\mathbf{p}) \land \mathbf{p} / \lambda \in \inc
					\land \pos_2(\mathbf{a},\lambda',\mathbf{p'}) \land \mathbf{p'} / \lambda' \in \inc
					\land \mathbf{p} / \lambda = \mathbf{p'} / \lambda'$
\IF {$\bigvee_{1 \le i,j \le 2} \phi_{i,j}$ is satisfiable}
\RETURN \TRUE
\ELSE  \RETURN \FALSE
\ENDIF
\end{algorithmic}
\end{algorithm}

\begin{theorem}
\label{algo2-is-a-decision-procedure-for-elrf}
Let $C$ be an SLC loop.
\textup{Algorithm~\ref{algo2}} decides the existence of an increasing
function $f$ and a linear function $\rho$ such that $\rho$ is an
eventual linear ranking function for $(C, f)$.
\end{theorem}

Exactly as in the previous section, if Algorithm~\ref{algo2}
returns \textbf{true} then we can extract an increasing function $f$,
a threshold $k$, and a linear function $\rho$.
We can also generalize the approach to the fully automated
detection of eventual affine ranking functions.

With respect to complexity,
Algorithm~\ref{algo2}  is not polynomial for two reasons.
In step 1, computing the set $\inc$ of linear increasing functions for $C$
requires elimination of existentially quantified variables.
In step 2, formula $\phi_{2,2}$  leads to a non-linear system
and we may have to check its satisfiability in step 10.
Although decidable, we are not aware of the existence of
polynomial algorithms for these problems.

\subsection{Verification}
\label{sec:verification}

Given $C$ an SLC loop, an associated increasing
function $f$, and a linear function $\rho$, we want to know whether
$\rho$ is a ranking function.
We can run Algorithm~\ref{algo1}, with the coefficients $\mathbf{a}$
fully instantiated. If needed, we can compute the threshold $k$
as explained in Section~\ref{sec:elrf-semi-detection}.
It follows that the verification problem is polynomial.

\subsection{Implementation}
\label{sec:implementation}

We  have implemented both algorithms in SICStus Prolog.
However, as $\phi_{2,2}$ of Algorithm~\ref{algo2}
leads to a non-linear system, we relaxed this formula to
\[
  \dec_2(\mathbf{a}, \lambda, \mathbf{p})
    \land \mathbf{p} / \lambda \in \inc
    \land \pos_2(\mathbf{a},\lambda',\mathbf{p'})
    \land \mathbf{p'} / \lambda' \in \inc,
\]
which is now linear. As shown in the following proposition, the existence
of an eventual linear ranking function (hence termination) is preserved,
but the associated increasing function is not linear.

\begin{proposition}
Let $C$ be an SLC loop and assume that
\(
  \dec_2(\mathbf{a},\lambda,\mathbf{p})
    \land \mathbf{p} / \lambda \in \inc
    \land \pos_2(\mathbf{a},\lambda',\mathbf{p'})
    \land \mathbf{p'} / \lambda' \in \inc
\)
is true.
Then there exists a non-linear increasing function $f$ such that
$\rho(\mathbf{x})= \mathbf{a} \mathbf{x}$
is an eventual linear ranking function for ($C$,$f$).
\end{proposition}

\begin{proof}
As $\dec_2(\mathbf{a},\lambda,\mathbf{p})  \land \mathbf{p}/\lambda \in \inc$
is true, there exists an increasing function $f_d$ and a rational
$k_d$ such that when the value of $f_d$ is beyond $k_d$, $\rho$
decreases.
Similarly, as
$\pos_2(\mathbf{a}, \lambda',\mathbf{p'}) \land \mathbf{p'}/\lambda' \in \inc$
is true, there exists an increasing function $f_p$ and a rational
$k_p$ such that when the value of $f_p$ is beyond $k_p$, $\rho$ is
non-negative.
Let $k = \max(k_p, k_d)$ and
$f(\mathbf{x}) = \min\bigl(f_p(\mathbf{x}), f_d(\mathbf{x})\bigr)$.
One readily checks that $f$ is a non-linear increasing function for $C$
and $\rho$ is an eventual linear ranking function for $(C, f)$.
\end{proof}

\section{Related Work and Experiments }
\label{sec:related-work-and-experiments}

As eventual linear ranking functions generalize linear ranking
functions, we focus on related work that goes beyond linear ranking
functions for SLC loops.
In order to appreciate the relative power of the different methods,
we report on the results obtained with  our algorithms on the loops
discussed in the papers where the other approaches were introduced.

The method proposed in \cite{ChenFM12} repeatedly divides the state
space to find a linear ranking function on each subspace, and then
checks that the transitive closure of the transition relation is
included in the union of the ranking relations.
As the process may not terminate, one needs to bound the search.
\cite{ChenFM12} also proposes a test suite, upon which we tested
our approach.
As expected, every loop \cite[Table~1]{ChenFM12} which terminates with
a linear ranking also has an eventual linear ranking.
Moreover, loops 6, 12, 13, 18, 21, 23, 24, 26, 27, 28, 31, 32, 35, and~36
admit an eventual linear ranking function (which is discovered without
using neither $\phi_{2,2}$ nor its relaxation).
These are all shown terminating with the tool of \cite{ChenFM12}.
On the other hand, loops 14, 34, and 38 do have a
\emph{disjunctive ranking function}
(following the terminology of \cite{ChenFM12}),
but do not admit an eventual linear ranking function.

\cite{GantyG13} shows how to partition the loop relation into
behaviors that terminate and behaviors to be analyzed in a subsequent
termination proof after refinement.
This work addresses both termination  and  conditional termination problems
in the same framework.
Concerning the benchmarks proposed in \cite[Table~1]{GantyG13},
loops 6--41 all have an eventually linear ranking function except for
loops 11, 14, 30, 34, and~38.

A method based on abstract interpretation for synthesizing
ranking functions is described in \cite{Urban13}.
Although the work contains no completeness result, the approach
is able to discover piecewise-defined ranking functions.

Finally, let us point out that the concept of \emph{eventual termination}
appeared first in~\cite{BradleyMS05ICALP,BradleyMS05VMCAI}.
The class loops studied in these works is wider but,
as the technique of~\cite{BradleyMS05VMCAI} relies on finite differences,
this approach is incomplete.
On the other hand, while \cite{BradleyMS05ICALP} is also based
on Farkas' Lemma, it seems
[A.~R.~Bradley, Personal communication, May 2013]
that the \emph{polyranking} approach cannot
prove, e.g., termination of the SLC loop
$p(x, y) \leftarrow x \ge 1, x' = y, y' = y-1, p(x', y')$, which admits
an eventual linear ranking function.

\section{Conclusion and Future Work}
\label{sec:conclusion-and-future-work}

We have proposed a definition of eventual linear ranking
function for SLC loops that strictly generalizes the concept
of linear ranking function.
We also defined two correct and complete algorithms for detecting such
ranking functions under different hypotheses.
The first algorithm shows that the mere knowledge
of the right increasing function allows checking the existence
or even synthesizing an eventual linear ranking function in polynomial time.
The second algorithm decides the existence of an
eventual linear ranking function in its full generality but
is not polynomial. We have also explained how to extend the algorithms
for deciding eventual affine ranking functions.
The algorithms admit a simple formulation as a constraint logic
program and have been fully implemented in SICStus Prolog inside the
BinTerm termination prover~\cite{SpotoMP10}.

It has to be noted that a nice property of the notion of eventual
(not necessarily linear) ranking function is its simplicity.
This is important when functions that witness termination have to be
provided (and/or understood) by humans.  This is the case when
annotating a C/ACSL program with loop variants \cite{BaudinCFMM+13}:
for the cases when a ranking function to be specified in a
\verb+loop variant+ clause is not obvious, one could extend ACSL with a
\verb+loop prevariant+ clause that allows the annotator
to indicate a candidate increasing function.
In the linear case, our first algorithm can efficiently decide
whether the two clauses constitute a termination witness.

On the other hand, there obviously are, as indicated
in Section~\ref{sec:related-work-and-experiments},
more complex classes of ranking functions and algorithms that allow
to establish the termination of SLC loops that do not admit an
eventual linear ranking functions.
A proper assessment of the relative merits of these approaches,
all extremely recent, requires an extensive experimental evaluation
that is one of our objectives for future work.

The verification of linear ranking functions for integer SLC loops, i.e.,
checking the satisfiability of
$c(\mathbf{x}, \mathbf{x}') \land \rho(\mathbf{x})  <  1 +  \rho(\mathbf{x}')$
and $c(\mathbf{x}, \mathbf{x}') \land \rho(\mathbf{x}) < 0$,
is an $\NP$-complete problem.
Concerning the existence of linear ranking functions,
as the Farkas' Lemma is not true for the integers,
the method presented in Section~\ref{sec:linear-ranking-functions} is not
valid.
The problem, which has been solved very recently in \cite{Ben-AmramG13},
is $\coNP$-complete, and the paper proposes an exponential-time algorithm.
Extending the present approach to integer SLC loops is another
interesting idea to consider for future work.

\paragraph{Acknowledgments}
We are grateful to Anthony Alezan, Aaron R. Bradley, \'Etienne Payet,
and some anonymous referees for their helpful comments.


\end{document}